\documentclass[10pt]{llncs}
\usepackage[utf8]{inputenc}
\usepackage{amsmath}
\usepackage[lofdepth,lotdepth]{subfig}

\usepackage{amssymb}
\usepackage{afterpage}

\usepackage{array}
\usepackage{stmaryrd}
\usepackage{graphicx}
\usepackage{color}
\usepackage{multirow}
\usepackage[normalem]{ulem}

\renewenvironment{proof}{\paragraph*{Proof.} }{\hfill\qed}

\newcommand{\pkcssharp}{%
  {\settoheight{\dimen0}{PKCS}PKCS\kern-.05em \resizebox{!}{\dimen0}{\raisebox{\depth}{\#}}}}
  
 \title{Asymptotic Analysis of Plausible Tree Hash Modes for SHA-3}
 \author{Kevin Atighehchi\inst{1} \and Alexis Bonnecaze\inst{2}}

\institute{Aix Marseille Univ, CNRS, LIF, Marseille, France\\
\email{kevin.atighehchi@gmail.com}\\
\and
Aix Marseille Univ, CNRS, I2M, Marseille, France
\\
\email{alexis.bonnecaze@univ-amu.fr}}

\begin{document}

\maketitle

\begin{abstract}
Discussions about the choice of a tree hash mode of operation for a
standardization have recently been undertaken. 
It appears that a single tree mode cannot address adequately 
all possible uses and specifications of a system. 
In this paper, we review the tree modes which have been proposed, we discuss their problems and propose solutions.
We make the reasonable assumption that communicating systems have different specifications and that 
software applications are of different types (securing stored content or live-streamed content).
Finally, we propose new modes of operation that address the resource usage problem for three representative 
categories of devices and we analyse their asymptotic behavior.
\end{abstract}



\section{Introduction}
\subsection{Context}

In this article, we are interested in the parallelism of cryptographic hash functions. Depending on the nature of the application, we  either seek 
\begin{itemize}
\item to first reduce the (asymptotic) parallel running time, and then the number of involved processors, while minimizing the required memory of an implementation using 
as few as one processor;
\item or to obtain an asymptotically optimal parallel time while containing the other computational resources.  
\end{itemize}
Historically, a cryptographic hash function makes use of an underlying function, denoted $f$, having a fixed input size, like a compression function,
a block cipher or more recently, a permutation \cite{BDPA13,Sponge2011}.  This underlying function, 
so-called {\it inner} function according to the terminology of Bertoni \emph{et al.} \cite{BDPA13}, is called 
iteratively on the message blocks in order to process a message of arbitrary length. When this mode of operation is sequential, it makes it difficult to 
exploit parallel architectures.
In fact, a sequential (or serial) hash function can only use 
Instruction-Level Parallelism (ILP) and Single Instruction Multiple Data (SIMD) \cite{GK12a,GK12b}, because the amount of computation which can be done in parallel 
between two consecutive synchronization points is too small.
However, some operating modes use hash functions as \emph{inner} functions and can exploit a particular tree structure either for \emph{parallelism} or \emph{incrementality}\footnote{The use of a balanced binary tree is particularly efficient to update the hash of an edited message. 
The change of one block in the message requires an update of the digest in logarithmic time.}  purposes. 
A (non-degenerated) tree structure  not only
allows for further use of SIMD instructions, but also enables the use of \emph{multithreading} in order to process in
parallel several parts of a message on different processors/cores.

The Internet is recognized as a network interconnecting heterogeneous computers, and this becomes particularly true with the advent
of the Internet of Things (IoT). 
The choice of a hashing mode depends on the nature of the 
computing platforms
that process
the message. 
According to \cite{HBPT15}, IoT devices can be classified in two categories, 
based on their capability and performance: \emph{high-end IoT devices} which regroup
single-board computers and smartphones, and \emph{low-end IoT devices} which are much more resource-constrained. Based on the memory 
capacity of its devices, this last category has been further subdivided \cite{rfc7228}
by the Internet Engineering Task Force (IETF). 
The time and space complexities to execute the tree mode are important for resource-constrained systems,
especially when this execution is sequential (\emph{i.e.} using a single processor).

In this article, we focus on the parallel efficiency of hash functions, depending on the chosen modes. 
What we investigate are the tree modes for an inner Variable-Input-Length function.
Instead of working at finite distance, 
we 
consider big-O asymptotic complexities  as it is usually done in 
other topics of algorithms and data structures
({\it e.g.} sorting and searching algorithms).
We choose to segment the parallel computers into 3~categories, each of which 
 can be mapped to a dedicated hashing mode:
 resource-constrained devices (no matter how low their CPU and RAM resources are), devices dedicated to critical applications (having abundant and possibly 
 specially-dedicated resources), and a last category which could constitute a middle-range. Such a mapping is valid only if communicating peers are of a same category.
 Otherwise, the use of a hashing mode dedicated to resource-constrained devices prevails.
 This work is not devoted to lightweight hash functions and the inner function that we take as example is based on the standardized \textsc{Keccak}~\cite{BDPA13}.
 However, there are no impediments to using some of our tree modes with an inner lightweight hash function. 

\subsection{Computational model and terminology}
Since a sequential 
hash function iterates a ``low-level'' primitive
on fixed-size blocks of the message (with maybe a constant number of added blocks for padding or other coding purposes), 
its running time is asymptotically linear in the message size. In terms of memory usage, such a function needs to store
only one hash state during its execution. The hash state size corresponds to the output width of the inner function.
For instance,
in the case of sponge-based hash functions like \textsc{Keccak}, 
the hash state size corresponds to the width of the underlying permutation.

Several conventions exist to describe tree structures in the context of parallel hashing. The first convention, denoted C1 and often used to deal with Merkle trees, 
consists in considering a node as the result of $f$ applied on the concatenation of its children. A leaf is the result of the inner function applied 
on an individual block of the message. A second convention, denoted C2, is a variant of C1 in which the leaves are simply the blocks of the message. The term 
\emph{hash tree} is used to refer to this type of tree, where the nodes are the $f$-images and the leaves are the blocks of the message.
The last encountered convention \cite{BDPV09,BDPV14_Suf}, denoted C3, considers the nodes as being the inputs of the inner function. Throughout this paper, 
unless otherwise specified, we use the convention~C3. Thus, a node is an $f$-input, while a chaining value is an $f$-image. 
The tree height, denoted $h$, corresponds to the number of levels, indexed from $1$ (for the base level) up to $h$
(for the final level). For the sake of simplicity, our results are derived by considering the nodes as containing 
either chaining values or message blocks, but not both. A level is said to be of arity $a$ if all its nodes contain exactly $a$ chaining values 
(or message blocks), except one node whose arity may be smaller (this is an ancestor node of the last message block processed). 

In order to establish  our complexity results, we use the classic PRAM (Parallel Random Access Machine) model of computation assuming the strategy EREW (Exclusive Read Exclusive Write).
Our results are presented in terms of big-O (symbol~$O$). Our proofs use the asymptotic equivalence 
(symbol~$\sim$) and sometimes the big-Theta notation (symbol~$\Theta$).

For a given set of parameters that characterizes the tree topology of a scalable tree mode, the ``ideal'' parallel running time corresponds 
to the time needed to compute the final digest when the number of processors is not \emph{a priori} bounded. 
Actually, this number of processors is a function of the message size and 
the chosen parameters. With the convention used for a node (\emph{i.e.} an $f$-input), 
it corresponds exactly to the number of nodes that embed message bits.
Such a definition makes the assumption that nodes embedding both message bits and chaining values bits position the former before the latter 
(\emph{i.e. a message part followed by one or more chaining values}, termed as \emph{kangaroo hopping} in~\cite{BDPV14_Sak}).
The term \emph{amount of work} refers to the total number of calls to the low-level primitive (\emph{e.g.} the permutation of \textsc{Keccak}) 
performed by the processors. It can be seen as the \emph{sequential} running time. We say that it is asymptotically optimal
if it is equal to the asymptotic running time---in $O(n)$---of a traditional (serial) hash function. As regards tree-based hash functions, 
the worst \emph{amount of work} is obtained when the node arities are small and the tree height unrestricted. For instance, the amount of work
to compute the root of a classic Merkle binary tree is at least twice that of a serial hash function.
For a parallel algorithm, the amount of work is lower than or equal to the product number of processors $\times$ parallel time. 
We say that the number of processors is optimal for the desired parallel time if this product is in~$O(n)$.

\subsection{Our contributions} 

After a concise state of the art on tree hash modes, we describe strategies for sequential and parallel implementations of tree hash modes, with particular attention to the memory consumption.
For a sequential implementation, we show that the memory consumption of a tree hash mode is asymptotically linear in the height of the produced tree, whatever its node arities.
We then propose several scalable tree modes using \textsc{Sakura} 
coding~\cite{BDPV14_Sak}. These modes 
address the memory usage problem, the parallel time and the 
required number of processors:
\begin{itemize}
 \item We show how to parameterize the tree topology and give 3 modes (4S, 5S and 6S) suitable for the hashing of (streamed or not) stored content.
 \item Then, we show that it is interesting to have a sequence of increasing levels arities, or even levels of increasing node arities. 
 This leads us to propose 3 modes (4L, 5L and 6L) suitable for the 
 hashing of live-streamed content. While at first glance this seems somewhat contradictory, we show that without knowing in advance the size of the message, 
 these~3 adaptive tree constructions actually lead to different asymptotic complexities.
 \item We discuss the way of decreasing the number of processors required to obtain the ideal (asymptotic) parallel running time.
 The proposed modes then become optimal with respect to the \emph{amount of work}.
 \item We give some guidelines for the use of interleaving.
 \item We make suggestions for generating a digest of arbitrary length using a parallel construction.
\end{itemize}
\subsection{Organization of the article} After  a brief survey on tree hash modes, Section~\ref{prel} contains 
background information
regarding hash functions and tree hash modes. We discuss their security, implementation strategies and their time-space efficiency. Using a 
parameterizable tree hash mode described in Section~\ref{parameterizable_mode}, we derive several modes 
addressing the memory usage problem,  the parallel time and the number of processors. 
In particular, Section~\ref{param_stored} gives parameters that produce tree topologies suitable for streamed stored
content. Then, parameters suitable for live-streamed content are given in Section~\ref{param_live}. 
Finally, we discuss in Section \ref{Interleave_discussion}
how we can conciliate scalability and interleaving, and concluding remarks are given in the last section.


\section{Overview of tree hash modes and motivations}\label{survey}
A tree hash mode uses an inner function $f$ to compute the hashes of nodes based on 
the values of their children. Depending on the target application, the result can simply be the final digest of the hash tree (\emph{i.e.} the hash of the root node), 
or all the computed $f$-images.
Tree hashing is due to Merkle and Damgård \cite{Dam90,Mer79} and has several applications: \emph{Post-Quantum Cryptography} with Merkle signatures, 
\emph{Incremental Cryptography}, \emph{Authenticated Dictionaries}, and the field we are concerned with here, \emph{Parallel Cryptography}.

Tree hash modes have been proposed in the SHA-3 candidates Skein \cite{FLSWBKCW09} and MD6 \cite{RABCDEKKLRSSSTY08}, and also in BLAKE2 \cite{ANWW13}. These 
tree modes are slightly parameterizable since the arity of the tree can be chosen. 
Better still, in Skein, the node arities are slightly more customizable: 
a parameter $\lambda_{in}$ indicates that the inner nodes (\emph{i.e.} nodes of level~$\geq 2$) are of arity $2^{\lambda_{in}}$  
and a parameter $\lambda_{leaf}$ indicates that the base level nodes each contain $2^{\lambda_{leaf}}$ message blocks. 
Skein, BLAKE2 and MD6 have also a parameter restricting the tree height.
If this last parameter has a too small value, the root node  can have an arity proportional to the size of the message.

 Bertoni \textit{et al.} \cite{BDPV09,BDPV14_Suf} give sufficient conditions for a tree-based hash function to ensure its indifferentiability
from a random oracle. 
They define the flexible \textsc{Sakura} coding \cite{BDPV14_Sak} which ensures these conditions, and enables any hash algorithm using it to be indifferentiable from
a random oracle, automatically. 
More particularly, if all the tree hash modes are compliant with this coding and operate the same inner function, (trivial) inter-collisions are avoided---or are 
merely reduced to collisions of the inner function.

They also propose  several tree hash modes for different usages. 
We can compare the efficiency of these algorithms using Big-O notation. 
For example, a mode---called Mode 1 in this article, ``final node growing'' in~\cite{BDPV14_Sak}, or ``unlimited fanout'' in~\cite{ANWW13}---can make use of a tree of height 2, defined in the following way:  the message is divided into fixed-size chunks which have to be hashed 
separately. The hash computations are distributed among the processors, and the concatenation of the resulting digests is sequentially hashed by a single processor. 
The advantages of this mode is its scalability (the number
of processors can be linear in the number of blocks) and its reduced memory usage when executed sequentially. Its drawback is its ideal running time which remains
linear in the message size. In Mode~2, the message is divided into as many parts (of roughly equal size) 
as there are processors so that each processor hashes each part, 
and then the concatenation of all the results is sequentially hashed by one processor. 
In order to divide the message into parts of roughly equal size, the size of the message is required as input to the algorithm, which limits its use to the hashing 
of stored (or streamed stored) contents.
Bertoni \textit{et al.} use an idea from Gueron~\cite{Gue14} to propose a variant (Mode~2L)
which still makes use of a two-level tree and a fixed number of processors, but this one interleaves 
the blocks of the message. This interleaving, which consists in distributing the message bits (or blocks) in a \emph{round-robin} fashion among $q$ clusters, 
has several advantages. It allows an efficient parallel hashing of a streamed message, 
a roughly equal distribution of the data processed by each processor in the first level of the tree (without prior knowledge of the message size), 
and finally a correct alignment of the data in the processors' registers (for SIMD implementations). The major drawback of this interleaving is that the memory consumption is 
$O(q)$ if the message bits have to be processed in the order of their arrival, 
no matter whether this tree hash function is sequentially implemented or not.
Finally, the classic binary tree, in Mode~3, offers the best ideal running time
but it consumes a lot of storage when executed by a single processor. 
Let $M$ be a message of $n$ blocks, each block being of fixed-length $N$. 
Table~\ref{Efficiency_StoA} compares the efficiency of these algorithms.

\begin{table*}
\centering
 \begin{tabular}{||>{\centering\arraybackslash}p{0.9cm}||>{\centering\arraybackslash}p{1.4cm}|>{\centering\arraybackslash}p{1.7cm}|>{\centering\arraybackslash}p{1.6cm}|>{\centering\arraybackslash}p{2cm}|>{\centering\arraybackslash}p{3.6cm}||} 
   \hline
    Mode & Live streaming & Memory (sequential) 
    & Number of processors & Parallel running time (ideal case) & Comments\\
    \hline
    \hline
    1  & \checkmark & $1$  & $n$ & $n$ & root of unlimited arity  \\
    \hline
    2S & - & $1$  & $q$ & \multirow{2}{*}{$n/q$} & \multirow{2}{*}{\shortstack{not scalable (but\\ reduced amount of work)}}\\
    2L & \checkmark & $q$  & $q$ & & \\
    \hline
    3 & \checkmark & $\log n$ & $n$ & $\log n$ & tree of unlimited height \\
    \hline
\end{tabular}
\vspace{0.2cm}\\
\caption{Asymptotic efficiency using Big-O notation of existing tree hash modes, 
where $n$ is the number of blocks of the message. 
Mode~2 is dedicated to a ``fixed'' number $q$ of processors. 
Its asymptotic efficiency is given without hiding the quantities $q$ and $1/q$ in the Big-O.}
\label{Efficiency_StoA}
\end{table*}

There was a debate \cite{luk13,kel14} about the way of standardizing tree hash modes. On the one hand, some  wanted
a single and simple tree hash mode allowing unrestricted depth\footnote{The term unrestricted can be misused. For some of these modes, 
there is a parameter defining the depth of the tree, and 
this one can be set large enough so that, in practice, it does not have any impact on the tree topology.} (like Skein \cite{FLSWBKCW09}, MD6 \cite{RABCDEKKLRSSSTY08} or 
BLAKE2 \cite{ANWW13}), with maybe several sets of parameters for the node arities.
These tree topologies are flexible and have a good potential parallelism, because they support live streaming, are scalable and allow 
a nice ideal speedup (in running time).
The problem is that, when its height is unbounded, such a tree brings a performance penalty for sequential execution, as the memory consumption and the
amount of \emph{work} (\emph{i.e.} computations) are much greater than for a serial (traditional) hash function. Note that the asymptotic efficiency of such a tree
is the same as that of Mode~3. However, if a parameter restricts its height (as allowed by Skein, MD6 or BLAKE2), 
its asymptotic (parallel) efficiency can fall into the case of Mode~1.

The choice of a tree hash mode, such as Mode~3, could be motivated by the potential speedup obtained in ideal conditions.
Its drawbacks are both the height of the produced tree structure which increases logarithmically with the size of the message and 
a substantial addition of computations. 
Most existing tree-based hash functions propose 
a parameter to limit the height of the tree, that we denote~$t$. If this parameter is set, there exists a message size threshold, denoted $l_t$, 
from which the final node will grow proportionally with the message size, thus leading to a performance penalty. 
Let us suppose that we have a new mode, 
denoted Mode~X, which, for message sizes exceeding this value, constructs a tree of height $t$ offering a better potential 
speedup than Mode~3. If the message size is known in advance and the choice of a tree hash mode is still motivated by the potential speedup, 
we have the following simple but interesting strategy: \emph{Choose $\mathrm{Mode~X}$ if the message size exceeds $l_t$, and $\mathrm{Mode~3}$ otherwise}.

\begin{figure*}[h!]
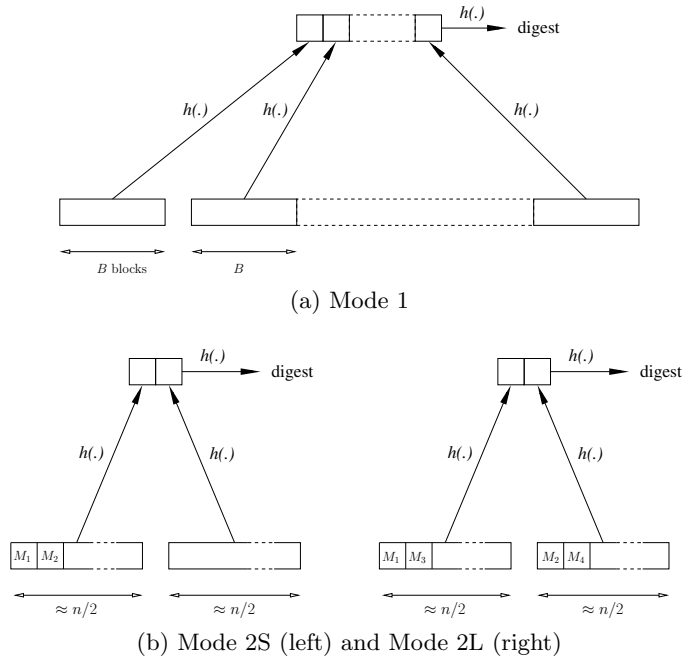

\begin{center}
    \subfloat[Mode 1]{
    \centering
    \scalebox{0.35}{
    \input{tree_2_levels_mode_1.pspdftex}
    }
    }
    
    \subfloat[Mode 2S (left) and Mode 2L (right)]{
    \centering
    \scalebox{0.35}{
    \input{tree_2_levels_mode_2.pspdftex}
    }
    }
\end{center}
\caption{Illustration of Mode 1 (at the top) and Mode 2 (at the bottom). At the top, the message is divided into chunks of fixed size $B$, 
while at the bottom it is divided into two chunks of roughly the same size ($\approx n/2$). In this example, Mode 2 is dedicated to the use of two processors.}
\label{illustration_1}
\end{figure*}

On the other hand, some  argue that there should be as many tree modes as application usages.
According to Kelsey~\cite{kel14}, there should be two standards: one standard for parallel hashing and one standard for tree hashing.
The standard for tree hashing would focus on trees of arbitrary (unrestricted) depth, with small node arities. These tree topologies 
are suitable for \emph{timestamping}, \emph{authenticated dictionaries} or \emph{Merkle signatures} (and their variants).
The standard for (fast) parallel hashing would focus on trees having a small height, because the evaluation of a hash function
should remain efficient on resource-constrained devices (having few memory and maybe a single processor). Indeed, as we will see later,
the memory consumption is linear in the tree height. Moreover, a small tree height means a reduced amount of \emph{work}. 
Since mid-2016, the standard for parallel hashing is available in NIST SP 800-185, with the specification of a \verb+ParallelHash+ function \cite{NIST800-185} 
adopting the tree structure of Mode~1.
At the same time, Bertoni \emph{et al.} \cite{BDPVV16} have proposed another variant of Mode~1 implemented with numerous optimizations, 
in particular the use of SIMD instructions, kangaroo hopping and a Round-Reduced version of \textsc{Keccak}.

Even for a client-server application, the memory consumption of a tree hash mode is a concern. Let us take the example of a 
\emph{Cloud Storage} application making use of Mode~3.
If, for checking their integrity, a storage server computes simultaneously the hashes of
$n$-block files uploaded from $v$ clients, the total memory consumption is then $O(v\log n)$.
That has to be compared to the $O(v)$ required by a Mode~1 based hash function or a serial hash function (\emph{e.g.} SHA-3). These differences should 
be taken into account when dimensioning the server.

Thus, the discussed modes are roughly the ones summarized in Table~\ref{Efficiency_StoA}.
Some illustrations are depicted in Figure~\ref{illustration_1} and Figure~\ref{illustration_2}.
Even if it is scalable and allows an optimal running time (in ideal conditions),  
Mode~3 seems to be left aside. It is just recommended for \emph{incremental hashing}, or for the other cited cryptographic algorithms and protocols.
Regarding 
the other proposed modes, it seems that a choice has to be made between scalability (Mode 1) 
and a reduced sequential part of the computation (the root node computation in Mode 2). In practice, for a small number of processors, 
this makes a difference. Asymptotically, in either case, the parallel time is still linear in the size of the message.

\begin{figure*}[h!]
\begin{center}
    \subfloat[Mode 3]{
    \centering
    \scalebox{0.35}{
    \input{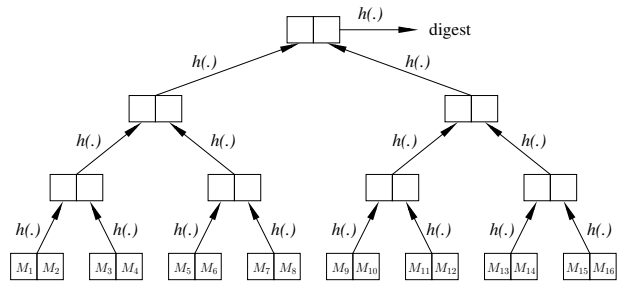}
    }
    }
    
    \subfloat[Variant of Mode 3 when restricted to 3 levels]{
    \centering
    \scalebox{0.35}{
    \input{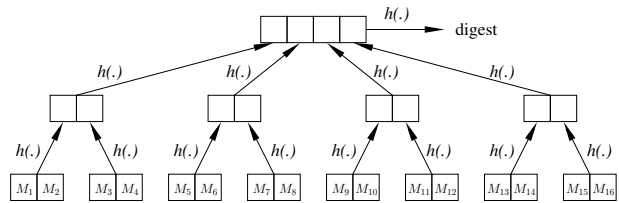}
    }
    }
\end{center}
\caption{Illustration of Mode 3 with and without height restriction, at the bottom and at the top respectively. At the bottom, the root node has an arity which 
grows proportionally with the size of the message.}
\label{illustration_2}
\end{figure*}

Mode 1 and Mode 2 have the advantage of requiring a constant memory consumption when executed sequentially. Nevertheless, under such a memory constraint, 
their ideal asymptotic parallel time is far from optimal. How can we reduce this time while leaving unchanged the space consumption? As regards Mode 3, the situation is 
reversed and we might wish to trade the logarithmic parallel time for a less decreasing function in order to decrease the memory footprint of a sequential implementation. 
The paper deals with these issues because it is unclear how to build tree toplogies allowing trade-offs between these two complexity criteria. 
Once a tree topology achieving a particular pair (number of processors, parallel time) has been characterized, the first component is not necessarily optimal 
with respect to the \emph{amount of work} and thus becomes another matter of interest.

\section{Preliminaries}\label{prel}

\subsection{Security}

Bertoni \textit{et al.} \cite{BDPV09,BDPV14_Suf}
give some guidelines to design correctly a tree hash mode $\tau$ operating an inner hash (or compression) function $f$. 
They define three sufficient conditions (\emph{message-completeness}, \emph{final-node-separability}, \emph{tree-decodability}) 
which ensure that the constructed hash function $\tau_f$, which makes use of 
an ideal 
inner function $f$,
is indifferentiable from an ideal hash function. They propose to use particular frame bits (\emph{i.e.} meta information bits)  
in nodes (\emph{i.e.} $f$-inputs) in order to meet these conditions. These frame bits characterize the type of node processed.

These conditions ensure that no weaknesses will be introduced when using the inner function. For instance, with \emph{tree-decodability}, an inner
 collision in the tree is impossible without a collision for the inner function.
Andreeva \textit{et al.} have shown in 
\cite{AMP10}
that a hash function indifferentiable from a random oracle
satisfies the usual security notions, up to a certain degree, such as pre-image and second pre-image resistance, 
collision resistance and multicollision resistance.

In the modes we propose, we use \textsc{Sakura} coding which is specified with an ABNF grammar~\cite{BDPV14_Sak}. 
\textsc{Sakura} enables any tree-based hash function using it to be automatically indifferentiable
from a random oracle, without the need of further proofs.
The coding used in a node depends on some information about it. For instance, if a node 
has children, the information about their number is encoded inside it using \textsc{Sakura}.

\subsection{Implementation strategies and complexities}

Since an inner (sequential) 
hash function processes a number of bits which is a multiple of a certain block size $N$, its time complexity behaves like a 
staircase function of this number of bits. 
We say that the time complexity of this iterated hash function $f$
for the operation $f(x)$ can be approximately described as an affine function of its input size
$l$ (in number of blocks)
whose coefficients depend on the  choice of the hash function and its parameters.
For instance, we can use a sponge construction \cite{BDPVsf07} like the \textsc{Keccak} algorithm \cite{BDPA13},
the new SHA-3 standard~\cite{SHA_standard,RPCKNN09}. 
The two important parameters of this sponge construction are the rate $r$ and the capacity $c$. This 
construction uses a permutation $P$ to process a state $S$ of $r+c$ bits at each iteration, and is divided into two phases: the absorbing phase
which processes the message blocks and the squeezing phase which generates the hash output.
The collision resistance and pre-image resistance strengths are related to the bit-size $c/2$.  Throughout the paper, we suppose the 
use of the \emph{extendable-output function} RawSHAKE256 which needs a capacity of 512 bits and which, according to the standard FIPS~202 \cite{fips202},  
can have a state size of 1600 bits and subsequently a rate of 1088 bits. For its use in tree hash algorithms, Bertoni \emph{et al.} \cite{BDPV14_Sak} suggest to set
its digest size equal to the capacity $c$.


\subsubsection{Memory usage.} The memory space used by the execution of a hash function can be divided into two quantities, the space used by the message to hash,
and the \emph{auxiliary space} used to execute the function on the message. The latter is particularly important for memory-constrained devices. Besides,
in the case of streaming applications, a message can be processed by a system as it arrives without being stored. In such a case, the total memory space used
is approximately reduced to the \emph{auxiliary space}. In this paper, we refer to the \emph{auxiliary space} when speaking of memory usage. 
A sequential hash function needs to store $\Theta(1)$ hash states in memory.
For evaluating $f$ on all the nodes of a tree of height $h$, a sequential implementation needs to store $\Theta(h)$ hash states in memory, regardless of the node arities. This 
memory consumption is due to a \emph{highest node first} strategy, which consists to apply $f$ on the highest node first. For a classic $k$-ary Merkle tree where $k$ is
a small number, this node is the highest node that has all its chaining values ready (\emph{i.e.} instantiated by $f$) to be processed. 
This strategy is particularly used in Merkle tree traversal techniques \cite{SZY04,BDS08}. For a 
tree having nodes of very high arity (\emph{e.g.} of arity $a$, 
possibly dependent of the message size), this strategy has to be changed in the following way:  start (or continue) to evaluate $f$ on the highest node 
that has $d$ chaining values
not yet consumed (with $d$ a constant much lower than~$a$). 
Thus, this number $d$ serves as a threshold value to trigger the (continuation/progression of the) processing by $f$ of the 
node. With such a variant, there is no need  
to wait until all the children be processed to process their parent node, and, as a consequence, there can at most be $d$ hash states in memory 
per level. An example is depicted Figure~\ref{highest_node_first_variant}.

\begin{figure}[t!]
\centering
\scalebox{0.6}{
\input{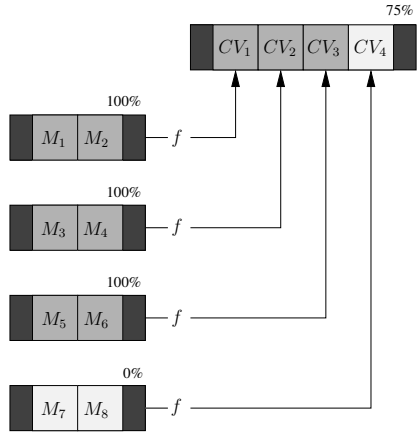}
}
\caption{Application of the variant of the \emph{highest node first} strategy. We represent the subtree of nodes containing the first 8 blocks of the message. 
Levels 1 and 2 are represented from left-to-right and nodes 
at a same level are represented vertically from bottom-to-top. In these nodes, black blocks contain possible frame bits and grey blocks message blocks 
or chaining values.
For reasons of simplicity, it is assumed that the cost to process frame bits is negligible and that one message block (or chaining value) is processed 
per unit of time. Percentages at the right corner of a node indicates the progression of its consumption by $f$, and dark grey blocks indicate 
which data have been used. We also suppose that
 $d=1$, meaning that the result of $f$ applied on a node is immediately used to advance the computation of its parent node. 
As a result, when 9 units of time have elapsed, the node at level 2 has been consumed by $f$ up to the three-quarter mark.}
\label{highest_node_first_variant}
\end{figure}

\subsubsection{Multithreading.} Having multiple processors/cores, we would like to use them to improve the computation time. 
We need to employ \emph{multithreading} to distribute,
by means of \emph{working threads}, the parallel computations among these processors. 
We assume here that we have $p$ processors and that we use a fixed thread pool containing $p$ threads 
(one thread per processor). 
Multithreaded implementations are very efficient if the threads do not need to communicate and/or synchronize, or
as little as possible. Synchronization delays are indeed very time-consuming. Depending on the scheduling strategy used, 
tree hashing can require a lot of synchronizations. Many situations have to be explored:
\begin{itemize}
 \item 
If the message to hash is already available (\emph{i.e.} locally stored on the system), we recommend to assign\footnote{Assigning the processing of a subtree 
 to a processor means that the latter is responsible for evaluating $f$ on all its nodes.} to each thread one of the $p$ biggest 
subtrees (of approximately the same size). If the tree topology makes it hard to find $p$ biggest subtrees covering the whole message, an other solution is 
to seek the highest level having a number of nodes $q \geq p$, and to assign to each thread $\lceil q/p \rceil$ of the $q$ subtrees rooted at these nodes. 
For the rest of the nodes, the method can vary. For instance, we can distribute the evaluations of $f$ on the remaining nodes 
as fairly as possible between the processors (at the cost of some synchronizations), or merely perform these computations sequentially. 
Such a strategy reduces greatly the number of synchronization points, and thus improves performances. 
Note that a thread processes its subtree as done by a sequential implementation, \emph{i.e.}, using
the \emph{highest node first} algorithm described above.
 \item If the message is received via a streaming system (no matter if it is stored or generated on the fly on the remote server), 
 the allocation strategy is necessarily fine-grained, with 
 a grain size depending on the bandwidth with which the message is received. Thus, compared to the previous case, such a parallel implementation 
 has to cope with more synchronization delays. Assuming that the link bandwidth is not a problem, we describe here a scheduling strategy for the 
 simple case where the tree arity $a$ is small. This one could be named \emph{higher level first} as it is similar to the strategy used 
 in a sequential implementation. 
 For a set of $p$ threads, at each level we use a buffer which can receive $p a$ values (they are message blocks or chaining values). 
 At any time, the threads are working at a same level of the tree and their goal is to fill up as soon as possible the highest level buffer.
 On a same level these threads have to apply $f$ on $p a$ nodes in order 
to move up and push $p$ chaining values in the buffer of the next level. Once these $p$ chaining values are computed, 
the buffer at the level below is flushed (\emph{i.e.} 
the $p a$ just processed chaining values are removed). If the current level occupied by threads is greater than $1$ and the lack of
resources on the level below prevents them from filling up the current buffer (\emph{i.e.} they cannot finish the computation of the  $p a$ values), 
then they return down to level 1, otherwise they continue, and so on. We let the reader deduce a termination phase for the end of the message, 
where buffers' contents of less than $p a$ values have to be processed. Note that depending on the throughput with which the message is received, 
this strategy could be adapted to process subtrees instead of nodes with the aim of increasing the amount of work done between two synchronization points. 
In other words, there could 
be a trade-off between the algorithm we just described and the one presented at the first point above. An other simple alternative to increase the amount
of work between two synchronization points is to consider buffers of $pak$ values, with a parameter $k > 1$ to select appropriately.
 \item There are possibly other situations in which the message blocks are not received in the order. They could be interleaved in a certain way for the need of 
 the hash function. In this paper, we do not discuss further more such a situation which seems unreasonable from a \emph{transport layer} standpoint 
 (\emph{e.g.}, in the OSI or DOD layered model for network protocol).
\end{itemize}
For processing a tree of height $h$, a parallel implementation using $p$
processors requires $O(ph)$ hash states in memory.

\subsubsection{SIMD implementations.} 
The single-instruction multiple-data (SIMD) units are present in a modern x86 processor or core. Well known instruction sets are
MMX, 3DNow!, SSE, SSE2, SSE3, SSSE3, SSE4, AVX and AVX2. 
These units can apply a same operation on several data simultaneously. They are thus very useful to compute several instances of 
a hash (or compression) function in parallel.  The size of SIMD registers
determines how many parallel hash (or compression) functions can be evaluated.

\subsubsection{Interleaving.} It is very useful for streamed content and serves two purposes: first, for multithreaded implementations, it 
allows the feeding of each processor as soon as possible. If each processor is responsible for the computation of a node, \emph{interleaving} allows the
distribution of the message bits/blocks among the nodes in a \emph{round-robin} fashion. Second, for SIMD implementation, it allows a perfect alignment
of the message blocks in the processor registers when loading the message sequentially. To the best of our knowledge, \emph{interleaving} has been proposed
to be used in a 2-level tree (Mode 2L of Table \ref{Efficiency_StoA}), under the name of $j$-lanes hashing \cite{Gue14}. In this proposition, the number
of nodes of the first level (\emph{i.e.} lanes) is a fixed parameter. Its use in trees of height $O(h)$, where $h$ is a function of the message size,
is an interesting question that will be discussed in Section \ref{Interleave_discussion}.

\section{A parameterizable tree hash mode}\label{parameterizable_mode}

We remind the reader that the number of hash states in memory corresponds, in the worst case, to the height $h$ of the tree.
Besides, if we denote by $u_i$ the biggest arity at level $i$, for $i=1 \ldots h$, the ideal parallel running time is (asymptotically) bounded by 
a linear function of $\sum_{i=1}^h u_i$.
This bound is indeed achieved when almost all the node arities of each level $i$ are equal to $u_i$, but
we will show that certain tree-growing strategies yield a parallel time linear in $h+u_1$.
In what follows, we present a tree hash mode using \textsc{Sakura} which, when the parameters are adequately chosen, produces modes offering interesting trade-offs between
memory consumption, parallel running time and number of processors. 
The parameters suitable for streaming stored content and streaming live content will be discussed
in Section \ref{param_stored} and Section \ref{param_live} respectively.

\subsection{Notations}

Let $f$ be a hash function which takes as  input a message $M$ of an
arbitrary length, and maps it to an output $f(M)$ of a fixed bit-length $N$. More concretely, $f$ can be
the intermediate \emph{extendable-output function} $\textrm{RawSHAKE256}$~\cite{fips202} with a digest size of $N=512$ bits.
Beforehand, the following notations are used in the description of the hashing mode:
\begin{itemize}
 \item The operator $\|$ denotes the concatenation.
 \item $\mathrm{I2OSP}(x,xLen)$ is a function specified in the standard \pkcssharp1 \cite{pkcs1v22}. It converts a non-negative integer, 
 denoted $x$, into a byte stream of specified length, denoted $xLen$. 
 \item $\{xLen\}$ is a single byte that represents the binary encoding of $xLen$. The arity $x$ of a node
 is encoded by using $\mathrm{I2OSP}(x,xLen)$ and $\{xLen\}$, where $xLen$ is appropriately chosen. 
 \item The encoding of the arity $x$ is defined
 as:
 $$\mathrm{enc}(x):=\mathrm{I2OSP}(x,\lfloor \log_{256} x \rfloor+1)\|\{\lfloor \log_{256} x \rfloor+1\}.$$
 \item $0^*$ indicates a non-negative number of bits $0$ to be used in an $f$-input for padding, for alignment purposes. 
 We will assume that this is the minimum number of bits $0$
 so that the bit-length of the $f$-input is a multiple of a word bit-length (\emph{i.e.} 32 or 64~bits).
 \item $I$ is the \emph{interleaving block size} of a node that determines how the message bits are distributed over its children. To the first child are given the first
 $I$ bits, to the second child is given the second sequence of $I$ bits and so on. After reaching the last child, we return to the first child, and so on. 
 A child node (of a node having an \emph{interleaving block size}) can
 have its own \emph{interleaving block size}, meaning that the bits it is responsible for are distributed over its children according to this attribute.
 This process can be repeated recursively if several generations of descendants have their own \emph{interleaving block size}.
 The notation $I_{\infty}$ means the absence of interleaving.
 \item $\{I\}$ represents two bytes that encode $I$. It is defined with a floating point representation, with one byte for the mantissa 
 and one byte for the exponent. For the absence of interleaving, $\{I_{\infty}\}$ is represented (using hexadecimal notation) with the coded mantissa 0xFF and the 
 coded exponent 0xFF.
 We refer to the \textsc{Sakura} specification \cite{BDPV14_Sak} for further information.
\end{itemize}

These encoded elements serve to identify the type of node and to delimit the embedded \emph{message hop}, \emph{chaining hops} or \emph{kangaroo hops}, in accordance with 
\textsc{Sakura} coding \cite{BDPV14_Sak}. These hops embed useful data (message blocks and chaining values).
The term \emph{payload} of a node refers to the list of its embedded message blocks and chaining values.

\subsection{The tree hash mode}

A hash tree can be characterised by
a two-dimensional sequence $(u_{i,j})_{i \geq 1, j \geq 1}$ of arities, where $u_{i,j}$ is the arity of node $(i,j)$, 
by assuming that the nodes are indexed at a same level with $j$, and from the base level to the root node with~$i$.
 In other words, the node $(1,1)$ 
  covers the first message blocks at the base level of the tree. In case all the nodes of level $i$
 are of same arity $u_i$ 
 (except the last node arity which may be smaller), 
 this tree can be characterised by
 a sequence $(u_i)_{i \geq 1}$ of arities.

Given a message $M$ of bit-length $|M|$ and a two-dimensional sequence 
of arities $(u_{i,j})_{i \geq 1, j \geq 1}$, SHAKE256 can be defined with the following tree hash mode using \textsc{Sakura}:
\begin{enumerate}
  \item Let $l_0=\lceil |M|/N \rceil$ and $M_0=M$. The quantity $l_0$ is the number of blocks of $M$, where the last block may be shorter than $N$ bits.
  We set $i=0$.
 \item We split $M_i$ into blocks $M_{i,1}$, $M_{i,2}$, ..., $M_{i,l_{i+1}}$ where:~\\ 
 ($1$) $l_{i+1}=\inf\left\{k\mid \sum_{j=1}^k u_{i+1,j} \geq l_i\right\}$;~\\ 
 ($2$)~$M_{i,j}$ with $j < l_{i+1}$ is $u_{i+1,j}N$ bits long, and $M_{i,l_{i+1}}$ may be shorter than $u_{i+1,l_{i+1}}N$ bits. ~\\
The node arities $(u'_{j})_{j \geq 1}$ of the current level are defined as follows:
\begin{equation*}
u'_{j}:=
\left\lbrace
\begin{array}{ccl}
u_{i+1,j} & \mbox{for} & j=1 \ldots l_{i+1}-1,\\
l_{i} - \sum_{v=1}^{l_{i+1}-1} u_{i+1,v} & \mbox{for} & j=l_{i+1}.
\end{array}\right.
\end{equation*}
 Then, we check certain conditions to apply \textsc{Sakura} coding correctly:
 \begin{itemize}
  \item If $i=0$ and $l_{i+1} \geq 2$, we compute the message
  $$M_{i+1}:= \bigparallel_{j=1}^{l_{i+1}} f\left(M_{i,j}\|110^*0\right).$$
  \item If $i=0$ and $l_{i+1} = 1$, we compute the message
  $$M_{i+1}:= f\left(M_{i,1}\|11\right).$$
    \item If $i > 0$ and $l_{i+1} > 1$, we compute the message
  $$\mkern-64muM_{i+1}:= \bigparallel_{j=1}^{l_{i+1}} f\left(M_{i,j}\|\mathrm{enc}(u'_{j})\|\{I_{\infty}\}\|010^*0\right).$$
  \paragraph*{Remark:} In certain tree-growing strategies, all the nodes at a same level can have the same arity,
  except the rightmost one, the arity of which may be lower. For SIMD implementation purposes, it is suggested 
  that the number of padding bits for the rightmost $f$-input be such that all the $f$-inputs, at this level of the tree, have same length.
  \vspace{0.1cm}
  \item If $i > 0$ and $l_{i+1} = 1$, we compute the message
  $$\mkern-35muM_{i+1}:= f\left(M_{i,1}\|\mathrm{enc}(u'_{1})\|\{I_{\infty}\}\|01\right).$$
 \end{itemize}
 
 \item We set $i=i+1$. If $l_i=1$, we return the hash value $M_i$. Otherwise, we return to step 2.

\end{enumerate}

\noindent
Let us consider a node, denoted $\mathrm{R}_*$, and its child nodes, denoted $\mathrm{R}_i$ for $i=1 \ldots k$, indexed in the order in which their corresponding chaining 
values are processed in $\mathrm{R}_*$. Let us denote by $T(\mathrm{R})$ the parallel time to process a node $\mathrm{R}$ and all its descendants. For the tree hash mode defined above, we
have
\begin{equation}
 \label{Parallel_time_def}
 T(\mathrm{R}_*)=\Theta \left(\max_{i=1 \ldots k}\{T(\mathrm{R}_i)+(k-i+1)\alpha\}\right),
\end{equation}
where $\alpha$ is the time to process a chaining value or message block.

\paragraph*{Case of a fixed arity at each level.}
Using the proposed mode, we remark that 
$$\lceil \lceil \cdots \lceil \lceil n/u_1 \rceil / u_2 \rceil \cdots \rceil/ u_i \rceil = \lceil n/(u_1u_2 \cdots u_i) \rceil$$ 
for a sequence of (strictly) positive integers $(u_j)_{j=1 \ldots i}$.
Consider a sequence of arities $(u_j)_{j \geq 1}$. At level $i$, there are exactly $\lceil n/(u_1u_2 \cdots u_i) \rceil$ nodes. 
If this sequence has an increasing number of terms greater than or equal to~$2$, then
there exists an integer $h>0$ such that
$\prod_{j=1}^h u_j \ge n$
where $n$ is the number of blocks of the message.
This ensures that we obtain a tree structure since, at level $h$, it remains a single node, the final (root) node. 
The problem is to find a sequence of arities $(u_j)_{j \geq 1}$ such that the tree height $h$ is $O(g_1(n))$ and
$\sum_{j=1}^h u_j$ is $O(g_2(n))$, where $g_1(n)$ and $g_2(n)$ are the desired complexities. Indeed, the memory usage
of a sequential implementation and the ideal parallel running time are related to these two quantities. Note that the ideal parallel running time
is accurately related to $\sum_{i=1}^h (a \cdot u_i+b)$ with $u_i \geq 1$, and $a$ and $b$ strictly positive integer constants such that $\alpha=a+b$.

\paragraph*{Other tree-growing strategies.} The above method does not apply and a case-by-case analysis is required. 
~\\

Note that this parameterizable mode is as simple as possible and that certain computational overheads are not avoided. In particular, all
the leaves are at the same depth. To decrease slightly the number
of calls to the underlying permutation, the concept of kangaroo hopping\footnote{The 
number of calls to the permutation can be slightly decreased by replacing the first chaining value of a node by the payload of the $f$-input (node) 
by which this chaining value was produced. The payload of a resulting node corresponds to a message block followed by one or more chaining values. 
By using \textsc{sakura}, it is augmented with adequate frame bits. The message block is then encoded in a message hop and, if kangaroo hopping is used once or several times, 
the chaining values are encoded in one or several chaining hops.} \cite{BDPV14_Sak} can be used. Such an optimization
reduces the computational costs of a mode at finite distance but does not change its asymptotic behaviour.
~\\

In the following two sections, we propose three new modes for streaming stored content (4S, 5S and 6S) and their variants for streaming live content (4L, 5L, 6L). 
Two of them can be seen as refinements of Mode~1 and Mode~3. In Mode~6, we optimize the tree topology subject to the constraint of an (asymptotically) optimal 
parallel time, while in Mode~4 we optimize the tree topology for a restricted memory consumption. The third one, called Mode~5, represents a balanced compromise. Table \ref{our_modes} summarizes the complexities of these modes.
\begin{table*}
\centering
 \begin{tabular}{||p{9mm}||p{15mm}|p{23mm}|>{\centering\arraybackslash}p{1.5cm}|>{\centering\arraybackslash}p{2.4cm}|>{\centering\arraybackslash}p{2.7cm}||} 
   \hline
    \centering{Our} \centering{modes} & \centering{Live}\newline \centering{streaming} & \centering{Memory usage (sequential)}
    & \centering{Number of processors} & Parallel running time (ideal case) & Priority \\
    \hline
    \hline
    \centering{4S}  & \centering{-} & \centering{\multirow{2}{*}{$1$}}  & \multirow{2}{*}{$n^{1-\epsilon}$} & \multirow{2}{*}{$n^{\epsilon}$} &  \scriptsize   \multirow{2}{*}{Memory consumption} \\
    \cline{1-2}
    \centering{4L} & \centering{\checkmark} & & & & \\
    \hline
    \centering{5S}  & \centering{-} & \centering{\multirow{2}{*}{$\frac{\log n}{\log \log n}$}} & \centering{\multirow{2}{*}{$\frac{n \log \log n}{\log^{1+\epsilon} n}$}} & \multirow{2}{*}{$\frac{\log^{1+\epsilon}n}{\log \log n}$} & \multirow{2}{*}{\scriptsize  None (trade-off)} \\
    \cline{1-2}
    \centering{5L}  & \centering{\checkmark} &   &  & & \\
    \hline
    \centering{6S}  & \centering{-} & \centering{\multirow{2}{*}{$\log n - \log \log n$}} & \centering{\multirow{2}{*}{$\frac{n}{\log n}$}} & \multirow{2}{*}{$\log n$} & \scriptsize  \multirow{2}{*}{Parallel time} \\
    \cline{1-2}
    \centering{6L}  & \centering{\checkmark} &   &  &  & \\
    \hline
\end{tabular}
\vspace{0.2cm}\\
\caption{Asymptotic efficiency (using Big-O notation) of our tree modes, where $n$ is the number of blocks of the message and  
$\epsilon$ is a positive constant, such that $\epsilon<1$. The \emph{ideal parallel time}
refers to the running time when the number of processors is not \emph{a priori} bounded.}
\label{our_modes}
\end{table*}

\section{Parameters for streaming stored content}\label{param_stored}

Can we find a tree of sublogarithmic height on which a hash function could rely to benefit of 
a logarithmic parallel time?
We can try to find a height $h$ such that $hn^{1/h} = C \log n$, where $C$ is a constant.
The rationale is that we could hope for a sublogarithmic height at the counterpart of any multiplicative 
constant factor for the logarithmic parallel time.
Trying to solve this equation on real numbers enables us to answer negatively to the question raised here.
We have the equation $\log h + \frac{\log n}{h} = \log (C\log n)$ that we can write in the form 
$\frac{- \log n}{h} = \log (\frac{h}{C\log n})$. Continuing, we can obtain the form
\begin{equation}\label{eq_question}
\frac{-\log n}{h}e^{\frac{-\log n}{h}}=\frac{-1}{C}.
\end{equation}
We note that the Lambert $W$ function solves the equations of type $xe^x=y$, where the solution for the unknown variable $x$ 
is $W(y)$. By using it for solving Equation~\ref{eq_question}, we get $-\frac{\log n}{h}=W(\frac{-1}{C})$. Since the domain of $W$ is $[-1/e, \infty[$,
there is no real solution if $C < e$. If $C=e$, we have $W(\frac{-1}{e})=-1$ 
and then the only solution is $x=\log n$. If $C > e$, we have $-1/C \in ] -1/e, 0 [$ and then $W(-1/C)$ takes two 
real negative values, denoted $w_0$ and $w_{1}$, evaluated using the two branches $W_0$ and $W_{-1}$ of the Lambert function. 
Then, Equation~\ref{eq_question} has two solutions $h_0=-\frac{\log n}{w_0}$ and $h_1=-\frac{\log n}{w_{1}}$.~\\

When we have to hash a stored file or a streamed stored media, its size is known in advance. 
Assuming a fixed (but possibly different) arity at each level of the tree,
we can then use this information to
define a finite sequence of level arities $(u_i)_{i=1 \ldots h}$. 
We give three interesting 
sets of parameters:
\begin{itemize}
 \item \textbf{Mode 4S.} $u_i=\lceil n^{\epsilon} \rceil$ for all $i \in \llbracket 1, h \rrbracket$ 
 and $h=\frac{1}{\epsilon}$, with a positive constant $\epsilon < 1$ such that
$1/\epsilon$ is a strictly positive integer. For instance $\epsilon=1/2$.
 \item \textbf{Mode 5S.} $u_1=\left\lceil \frac{\log_2^{1+\epsilon}n}{\log_2 \log_2 n} \right\rceil$, 
 $u_i=\left\lceil \log_2^{\epsilon} n \right\rceil$ for all $i \in \llbracket 2, h \rrbracket$, and 
 $h=\left\lceil \frac{\log_2 (\lceil n/u_1 \rceil)}{\epsilon \log_2 \log_2 (\lceil n/u_1 \rceil)} \right\rceil + 1$, with a positive 
 constant $\epsilon < 1$.
 \item \textbf{Mode 6S.} $u_1=\lceil \log_c n \rceil$, $u_i=c$ for all $i \in \llbracket 2, h \rrbracket$ with an integer constant $c>1$,
 and $h=\left\lceil \log_c\left( \frac{n}{u_1} \right) \right\rceil + 1$.
\end{itemize}

Some examples of trees 
resulting from the application of Mode 4S and Mode 6S on small messages are illustrated in Figure \ref{Mode_4S_example} and Figure \ref{Mode_6S_example}.

\begin{figure}
\centering
\scalebox{0.6}{
\input{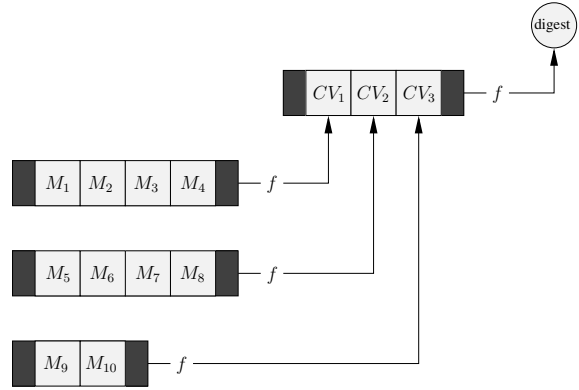}
}
\caption{Tree of nodes used by $f$ to compress a message of 10 blocks, assuming the use of Mode 4S 
with the parameter $\epsilon=1/2$. Levels 1 and 2 are represented from left-to-right and nodes 
at a same level are represented vertically from bottom-to-top. In these nodes, black blocks contain frame bits and light grey blocks message blocks 
or chaining values.}
\label{Mode_4S_example}
\end{figure}

\begin{figure}[t]
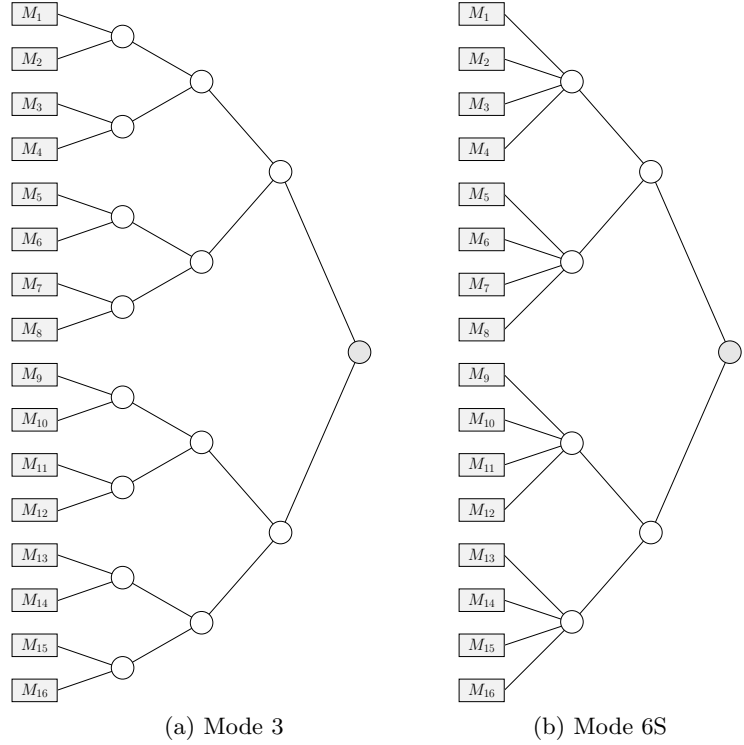

\begin{center}
    \subfloat[Mode 3]{
    \centering
    \scalebox{0.3}{
    \input{unrestricted_height_tree_2.pspdftex}
    }
    ~~~~~~~
    }
    \subfloat[Mode 6S]{
    \scalebox{0.3}{
    \input{unrestricted_height_tree_1.pspdftex}
    }
    }
\end{center}
\caption{Comparison of trees produced by Mode 3 and Mode 6S when applied on a message of 16 blocks. This compact representation uses convention C2 where nodes, depicted by circles,
are $f$-images. Mode 6S conserves the asymptotic parallel time of Mode 3 but reduces the number of processors, the memory usage and the amount of work (due 
to the reduced number of chaining values).}
\label{Mode_6S_example}
\end{figure}

\begin{theorem}
There are 3 tree hashing modes having the following efficiency complexities:
\begin{itemize}
 \item Mode 4S has an ideal parallel running time in $O(n^{\epsilon})$ by using only $O(n^{1-\epsilon})$ processors, 
 for a sequential memory consumption of $O(1)$ hash states. 
 \item Mode 5S has an ideal parallel running time in $O\left(\frac{\log^{1 + \epsilon} n}{\log \log n}\right)$ by using
 only $O\left(\frac{n \log \log n}{\log^{1 + \epsilon}n}\right)$ 
 processors, for a sequential memory consumption of $O\left(\frac{\log n}{\log \log n}\right)$ hash states. 
 \item Mode 6S has an ideal parallel running time in $O( \log n)$ by using only $O( \frac{n}{\log n} )$ 
 processors, for a sequential memory consumption of $O(\log(\frac{n}{\log n}))$ hash states. 
\end{itemize}

\end{theorem}

\begin{proof}
We examine the modes, in terms of parallel running time, number of processors and (sequential) memory consumption, one after the other:
\begin{itemize}
 \item Mode~4S is consistent since $\lceil n^{\epsilon}\rceil^{1/\epsilon} \geq n$. Its asymptotic parallel running time
 is clearly $O(n^{\epsilon}/\epsilon)$ and the height of the tree is $O(1/\epsilon)$, where $\epsilon$ is a constant.
 The number of involved processors is $n/\lceil n^{\epsilon} \rceil=O(n^{1-\epsilon})$.
  \item With Mode~5S, the parallel running time is obtained by construction. Indeed, we consider a tree covering $n'$ blocks and 
  seek the smallest $h'$ such 
  that $h'n'^{1/h'}= \frac{\log^{1+\epsilon}n'}{\epsilon \log \log n'}$. We have:
  $$n'^{1/h'}=\frac{1}{h'} \frac{\log^{1+\epsilon}n'}{\epsilon \log \log n'} \Leftrightarrow \frac{-\epsilon \log \log n'}{\log^{\epsilon}n}=\left(\frac{-1}{h'}\log n'\right) e^{\frac{-1}{h'}\log n'},$$
  and it follows that $\frac{-1}{h'}\log n'=W\left(\frac{-\epsilon \log \log n'}{\log^{\epsilon}n'}\right)$. Thus, we have 
  $$h=\frac{-\log n'}{W(-(\epsilon \log \log n') \log^{-\epsilon}n')}.$$ 
  Since 
  $-\epsilon \log \log n' \log^{-\epsilon}n'$ is negative and tends to $0$, we can evaluate $h'$ thanks to the two branches $W_{-1}(.)$ or $W(.)$ of the Lambert function. 
  We choose to use $W_{-1}(.)$ since we have the asymptotic approximation 
  $$W_{-1}(x)= \Omega\left(\log(-x)-\log(-\log(-x))\right)$$ 
  when $x \in [ -1/e, 0 [$.
  Consequently, $h'$ is $O(\frac{\log n'}{\epsilon \log \log n'})$. 
  Thus, we could consider a tree of height $h'=\left\lceil \frac{\log_2 (n')}{\epsilon \log_2 \log_2 (n')} \right\rceil$
  and of arity $a' = \lceil n'^{1/h'} \rceil \approx \lceil \log_2^{\epsilon} n' \rceil$.
  We now enlarge the tree by adding to it, at the bottom, 
  one level of $n'=\lceil n/u_1 \rceil$ nodes of arity $u_1=\left\lceil \frac{\log^{1+\epsilon}n}{\epsilon \log \log n} \right\rceil$.
  The height of the new tree is $h = h' + 1 = O(\frac{\log n}{\epsilon \log \log n})$ and the sum of 
  level arities is $u_1 + (h-1)a' = O(\frac{\log^{1+\epsilon}n}{\epsilon \log \log n})$. 
  We let the reader checking the consistency of the mode. 
  \item Mode~6S is consistent since $u_1 c^{\lceil \log_c(n/u_1) \rceil} \geq u_1 c^{\log_c n} u_1^{-1} \geq n$. Its parallel time
  is $\lceil \log_c n \rceil + c\left\lceil \log_c\left( \frac{n}{\lceil \log_c n \rceil}\right) \right\rceil = O(\log_c n)$ and $h$ 
  is $O\left(\log_c\left( \frac{n}{\log_c n} \right)\right)$.
\end{itemize}
\end{proof}

\paragraph*{Remark.} 
One way to reduce the number of processors when applying a mode is to use a rescheduling technique\footnote{A rescheduling consists in
distributing the tasks of a parallel step to a reduced set of processors whose number is lower than usual (and thus may be lower than the number of tasks). 
This is performed at each parallel step in the hope of getting the best parallel time for such a number of processors. 
The underlying data structure formed by the dependencies between the tasks 
(\emph{i.e.} the algorithm) is unchanged.} at runtime.
However, there is a benefit in
reworking the tree for reducing this number and also some overheads.
If a mode uses a tree topology allowing a given asymptotic parallel time,
a variant reducing the number of processors required
while conserving (asymptotically) the aforementioned parallel time exists.
Let $T(n)$ be the ideal parallel time of a mode having $a=\alpha(n)$ as fanout for all the nodes. 
We consider the following tree topology:
The first level has approximately $n/T(n)$ nodes, each containing $T(n)$ message blocks, except maybe the rightmost one, where the number of message blocks can be smaller.
An $a$-ary tree is constructed on top of these $n/T(n)$ nodes. Thus, the first level is of arity $T(n)$ while the others are of arity $a$. The nodes are then formatted 
using \textsc{Sakura} coding in accordance with this tree representation.
The overall parallel running time to compute the digest is
approximately $T(n)+T(n/T(n))$, by using only $n/T(n)$ processors. The computational overheads in space and amount of work are reduced 
against a parallel time almost doubled.


\section{Parameters for streaming live content}\label{param_live}

In this section, we discuss the parallel operating modes which do not require the message size as input. 
These modes are scalable and can process the message
as it is received. They are  essential for any application making use of live streaming.


\paragraph*{Can we benefit from increasing arities at a same level of the tree?} We respond to this question with a sketch of
tree mode which minimizes the tree height and allows the processing of live-streamed content.
We suppose that we have $k+1$ processors denoted $P_0$, $P_1$, $P_2$, ..., $P_k$. Processors $P_1$, ..., $P_k$ each compute
the hash of a distinct message part\footnote{A message part consists of a certain number of blocks.}. These chaining values are denoted $cv_1$, ..., $cv_k$. 
Figure~\ref{Mode_4L_example} depicts a tree produced by this mode.
For the sake of simplicity, we suppose that \emph{kangaroo hopping} is not used. 
Processor $P_0$ then merely collects and processes 
these chaining values as soon as they are evaluated. We suppose that the processors start their computations at the same time 
($P_0$ is delayed since \textit{kangaroo hopping} is not used). It appears that during the processing of a chaining value $cv_i$ by $P_0$, Processor $P_{i+1}$
can process one or more message blocks. Thus, the message parts processed could be of increasing lengths without producing 
a significant impact on the running time of $P_0$. Let us suppose the following complexities:
\begin{itemize}
 \item the running time of $P_i$ is $ai+b$ where $a$ and $b$ are positive constants, for $i \geq 1$.
 \item the running time of $P_0$ is related to the number of chaining values. For instance, it is $a'k+b'$ where $a'$ and $b'$ are positive constants. The waiting time 
 for $cv_1$ is supposed to be taken into account.
\end{itemize}
We have $\max\{a'k+b', \max_{1 \leq i \leq k} \{ai+b\}\} = O(k)$, where $k$ is the smallest integer such that $\sum_{i=1}^k i \geq n$. The total parallel running time
is then in $O(n^{1/2})$. The following tree mode, denoted Mode~4L, generalizes this for a given constant tree height.

\afterpage{
\begin{figure}[t!]
\centering
\scalebox{0.6}{
\input{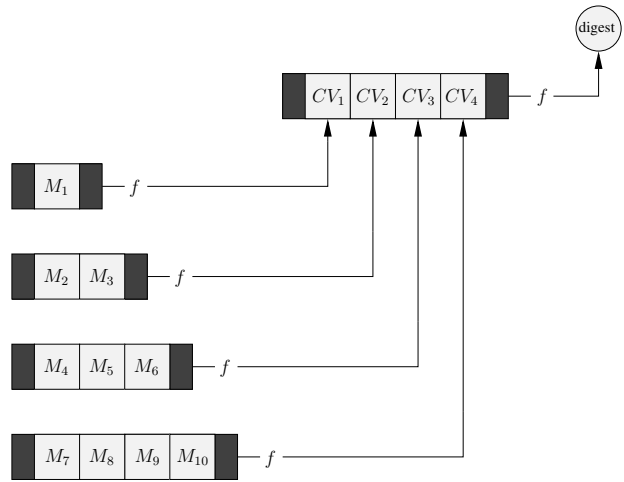}
}
\caption{Tree of nodes used by $f$ to compress a message of 10 blocks, assuming the use of Mode 4L for a minimized height.
Level 1 is composed of the nodes embedding the message blocks, while Level 2 is a single node embedding the chaining values.}
\label{Mode_4L_example}
\end{figure}

\begin{figure}[b!]
\centering
\scalebox{0.3}{
\input{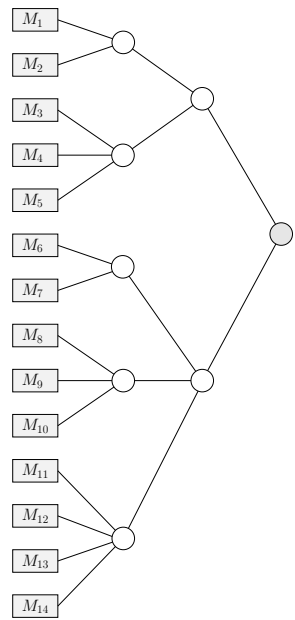}
}
\caption{Compact representation of a tree produced by Mode~4L for a 14-block message and a tree height $h=3$.}
\label{Mode_4L_examplebis}
\end{figure}
}

\paragraph{Mode 4L.} The underlying tree can be defined recursively until we obtain a given tree height $h$. Initially (before tree pruning), the root node is 
assumed to be of arity~$\infty$. Hence, 
a node of arity $k$ has its leftmost child node of arity $2$, its second child node from the left of arity $3$, ..., its $i$-th child node from the left of
arity $i+1$, ..., its rightmost child node of arity $k+1$. A node arity can also be computed directly. 
For instance, given a tree height $h = 4$, the node arities $u_{i,j}$ for $h \geq i \geq 1$ and $j \geq 1$ grow as follows:
\[u_{i,j}=\begin{cases} 
\infty & \mbox{if } i=4, \\ 
j+1 & \mbox{if } i = 3\ \mbox{or } (i < 3\ \mbox{and } j \leq 2),\\
2 + (j\mod{(k_1^2+k_1)/2}) & \mbox{if } i = 2\ \mbox{and } j > 2,\\
2 + (j\mod{(k_2^3+3k_2^2+2k_2)/6}) & \mbox{if } i = 1\ \mbox{and } j > 2,\\
\end{cases}
\]
where $k_1=\left\lfloor \frac{-1+\sqrt{1+8j}}{2} \right\rfloor$ and $k_2$ is the integer part of the positive solution of $k_2^3+3k_2^2+2k_2-6j=0$. A example
is depicted in Figure~\ref{Mode_4L_examplebis}.

\begin{theorem}
Mode~4L has an ideal parallel running time in $O\left( n^{\epsilon}\right)$ by using
only $O\left( n^{1-\epsilon} \right)$
processors, and requires  $O\left(\frac{1}{\epsilon}\right)=O(1)$ hash states in memory to process the message 
with a single processor, with $\epsilon=1/h$.
\end{theorem}

\begin{proof}
Let $a_{br}$ denote the arity of the righmost node at the base level, and let $h$ denote the height of the tree.
According to Equation~(\ref{Parallel_time_def}), the parallel time to compute the digest using a tree produced by Mode~4L is affine in $a_{br} + h$.
We suppose a tree of constant height $h=1/\epsilon$ where $\epsilon < 1$. Once we have pruned the unneeded branches in the tree, 
the root node has a finite arity denoted $x$. Considering a $h$-level tree produced by Mode~4L under the constraint of a $x$-ary root node,
the maximum number of blocks, denoted $n_{max}(x)$, covered by this tree is
\begin{eqnarray*}
\label{iter_sum}
\sum_{i_{h-1}=2}^{1+x}\ \sum_{i_{h-2}=2}^{1+i_{h-1}}\ \cdots\ \sum_{i_{2}=2}^{1+i_{3}}\  \sum_{i_1=2}^{1+i_{2}}\ i_1 &
= & \sum_{i_{h-1}=h}^{x+h}\ \sum_{i_{h-2}=h-1}^{i_{h-1}}\ \cdots\ \sum_{i_{2}=3}^{i_{3}}\  \sum_{i_1=2}^{i_{2}}\ i_1\\
& = & {x + h - 2 \choose h - 1} + {x + h - 1 \choose h}.
\end{eqnarray*}
By construction of the tree, $a_{br} \leq x+h$. Hence, we would like to express $x$ as a function of $n$. Given a message size $n$, 
we then seek the minimum $x$ satisfying $n_{max}(x) \geq n > n_{max}(x-1)$. Since $n_{max}(x)$ is polynomial in $x$, we have $\frac{n_{max}(x)}{n_{max}(x-1)} \underset{x\to \infty}{\longrightarrow} 1$, 
and since $n_{max}(x)>0$, it follows that $n_{max}(x) \sim n$. We can then seek $x$ such that 
$n_{max}(x) \sim n$. 
It can be shown that 
$$n_{max}(x) \sim \frac{x^h}{h!}.$$
Thus, we have $x \sim (n\, h!)^{1/h} \sim (2\pi h)^{1/(2h)}\ \frac{h}{e}\ n^{1/h}$ using the Stirling formula. Since $h=1/\epsilon$
is constant, we deduce that $a_{br} + h = O(n^{\epsilon})$.
The number of involved processors corresponds to the number of base level nodes, \emph{i.e.} the quantity
$$\sum_{i_{h-1}=2}^{1+x}\ \sum_{i_{h-2}=2}^{1+i_{h-1}}\ \cdots\ \sum_{i_{2}=2}^{1+i_{3}}\  i_2.$$
Given the estimated value $x = O(n^{1/h})$, it follows by composition of functions that this quantity is in $O(n^{\frac{h-1}{h}})$.
\end{proof}

\paragraph{Mode 5L.}The arities $u_{i,j}$ of the nodes $(i,j)$ for $i \geq 1$ and $j \geq 1$ grow as follows:
$$u_{i,j}=\begin{cases} 
\left\lceil \frac{\log^{1+\epsilon}_c(c+j)}{\log \log (c+j)} \right\rceil\ &  \mbox{if } i = 1\ \forall j \geq 1, \\ 
\lceil \log^{\epsilon}_c(c+j) \rceil\ &  \forall i > 1\ \forall j \geq 1,
\end{cases}$$
with an integer constant $c > 1$ and a positive constant $\epsilon < 1$.

\begin{theorem}
Mode 5L has an ideal parallel running time in $O\left(\frac{\log^{1 + \epsilon} n}{\log \log n}\right)$ by using
only $O\left(\frac{n \log \log n}{\log^{1 + \epsilon}n}\right)$
processors, and requires  $O\left(\frac{\log n}{\log \log n}\right)$ hash states in memory to process the message 
with a single processor.
\end{theorem}

\begin{proof}
The number of involved processors (or base level nodes) is the smallest $k_1$ such that 
$\sum_{j=1}^{k_1} \left\lceil \frac{\log^{1+\epsilon}_c(c+j)}{\log \log (c+j)} \right\rceil\ \geq n$.
Because this sum increases slowly, it can be shown that $\sum_{j=1}^{k_1} \left\lceil \frac{\log^{1+\epsilon}_c(c+j)}{\log \log (c+j)} \right\rceil\ \sim n$
and $\sum_{j=1}^{k_1} \left\lceil \frac{\log^{1+\epsilon}_c(c+j)}{\log \log (c+j)} \right\rceil \sim {k_1} \frac{\log^{1+\epsilon} k_1}{\log \log k_1}$.
We find the expected number of processors $k_1(n) \sim \frac{n \log \log n}{\log^{1+\epsilon}n}$.
Hence, we show that the height of the tree belongs to $O\left(\frac{\log n}{\log \log n}\right)$. 
We denote by $k_i$ the number of nodes at level $i \geq 1$. We have $k_2 \sim \frac{k_1}{\log^{1+\epsilon} (k_1)}$ and, more generally,
$k_i \sim \frac{k_1}{\log^{(1+\epsilon)i} (k_1)}$. We seek $i$ such that $k_i \sim 1$, \textit{i.e.} such that ${\log^{(1+\epsilon)i} (k_1)} \sim k_1$.
Applying the log function to both sides gives $i \sim \frac{\log k_1}{(1+\epsilon)\log \log k_1} \sim \frac{\log n}{(1+\epsilon)\log \log n}$.
Finally, the height $\times$ average node arity product gives the expected parallel 
time.
\end{proof}

\paragraph{Mode 6L.} The arities $u_{i,j}$ of the nodes $(i,j)$ for $i \geq 1$ and $j \geq 1$ grow as follows:
$$u_{i,j}=\begin{cases} 
\lceil \log_c(c+j) \rceil\ & \mbox{if } i=1, \\ 
c & \forall i > 1\ \forall j \geq 1,
\end{cases}$$
where $c > 1$ is an integer constant.

\begin{theorem}
Mode 6L has an ideal parallel running time in $O(\log n)$ by using only $O(\frac{n}{\log n})$ processors, 
and requires $O(\log \frac{n}{\log n})$ hash states in memory to process the message with a single processor.
\end{theorem}

\begin{proof}
Let us first look at the number of base level nodes.
For the sake of simplification, we seek the lowest $k$ such that $\sum_{j=3}^k \lceil \log_2 j \rceil \geq n$, which is
an equivalent problem. Not that, since $\sum_{j=3}^k \lceil \log_2 j \rceil/(\sum_{j=3}^{k-1} \lceil \log_2 j \rceil) \underset{k\to \infty}{\longrightarrow} 1$,
we have $\sum_{j=3}^k \lceil \log_2 j \rceil \sim n$. 
Asymptotically, we then seek $k$ such that $k \log k \sim \sum_{i=2}^k \log i \sim n$.
The solution meets the equivalence $\frac{n}{\log n} \sim \frac{k \log k}{\log(k) + \log \log(k)} \sim k$, which leads to the expected number of processors.
The average node arity of the first level is $\Theta(\log n)$ and we construct an $c$-ary tree on top of $\Theta(n/\log n)$ base level nodes.
The total parallel time is then $O(\log n)$ and the height $h$ of the tree is $\lceil \log_c(k) \rceil + 1 = O(\log(n/\log n))$.
\end{proof}

It is shown in Appendix \ref{appendix_param} that resource trade-offs are also possible using tree-growing strategies in which the arity does not vary in a same level 
but increases when getting close to the root node (\emph{i.e.} level-by-level).

\section{Conciliating interleaving and scalability}\label{Interleave_discussion}

In this section, we discuss approaches for hybrid Multithreading-SIMD implementations and the way to handle block interleaving.

\paragraph{Case of a fixed arity at each level.}
The interleaving of the message bits/blocks is an interesting feature if it is applied to nodes having a low number of children.
For instance, in Mode 2L, it is applied to a single node having $q$ children, and we can think that $q$ is small enough so that 
the memory space consumption is small as well. Even though \textsc{Sakura} allows the coding of an interleaving parameter for each node, we can wonder
if it is judicious to use it for all nodes in the tree. 
Let us imagine that our hash function is based on a tree structure with a logarithmic height and a constant arity (\emph{e.g.} all nodes are of arity 4).
In fact, with the assumption that message bits/blocks are processed in sequence, 
using interleaving for all nodes leads
to a memory consumption in the order of the size of the message. It is then preferable to use it for the nodes of a 
single level (or of a constant number of levels) in the tree.

Interleaving is bounded to a processor architecture, that is, it is difficult to imagine an interleaving mechanism 
that could be, in some sense, scaled on different architectures (with different register sizes). Having this in mind, we can still imagine an operating mode
that could be both bounded to a single architecture and scalable from the multithreading standpoint. In order
to use interleaving in the 
tree structures having nodes of high arity (possibly dependent of the message size),
we propose to add a single level of small and constant arity nodes. 
The interleaving parameter would be set\footnote{A value different from $\{I_{\infty}\}$ is used.} for all nodes of this level (and only this one).
The location where this level is inserted in the tree depends on the type of application (stored or live-streamed content).
In the case of stored content, this level would be inserted near the root. In the case of live-streamed content, this level would probably be inserted near
the base level in order to allow both the use of multithreading and SIMD instructions. In fact, in this second case, the most adequate location of the insertion
depends on the bandwidth with which the message is received. It should be inserted at the highest level\footnote{Considering 
a constant bandwidth and an infinitely large message, this level is clearly near the base level.} which allows the correct feeding of 
SIMD instructions executed in parallel by the processors.
Note that, in either case, the asymptotic complexity remains unchanged by this modification.

\paragraph{Case of increasing arities inside a level.} To evaluate several times the \emph{inner function} using SIMD instructions, it is desirable to 
have messages of same length. This requires another type of changes to the tree topology: One possibility is to increase 
the number of nodes of the base level by a (constant) multiplicative factor, thus leaving the height of the tree unchanged. 
We do not give a general method and
by way of example (see Figure~\ref{Mode_4L_interleaving_example}), we show how Mode~4L can be changed to support both 
multithreading and SIMD implementation, with or without block interleaving. 
We suppose that the cores have SIMD instructions operating on wide registers 
(\emph{e.g.} AVX2 with 256-bit YMM registers). For instance,
for exploiting a 4-way SIMD architecture, we multiply by $4$ the consecutive number of $i$-ary nodes at the base level and regroup them as quadruplets. 
We have a group of $4$ nodes of arity~$1$, followed by a group of $4$ nodes of arity $2$, followed by a group of $4$ nodes of arity $3$, and so on. 
If interleaving is desired, it is applied inside parts of increasing length (corresponding to the number of blocks
processed by a quadruplet of nodes), \emph{i.e.} the first group is responsible for the first part of $4$ blocks, ..., the $j$-th group is responsible 
for the $j$-th part of $4j$ blocks, and so on. 
Each part is divided into slices (64-bit qwords) that are distributed in a round-robin fashion among the nodes of a same group. Let us suppose
a message divided in qwords, \emph{i.e.} $M=s_1s_2 \ldots s_k$ where $|s_i|=64$, and let us denote $\textrm{Node}_{2,i}$ the $i$-th node of the second group. 
In this group, the slices are interleaved in the following way: $\textrm{Node}_{2,i}$ is composed of the slices $s_{32+i}$, $s_{36+i}$, $s_{40+i}$, ..., $s_{92+i}$.
Note that a slight update of \textsc{Sakura} coding is required to support interleaving in a small group of child nodes. 
In addition to the parameter $I=64$, the group size $n_I=4$ should also be encoded.

\begin{figure}[h!]
\centering
\scalebox{0.6}{
\input{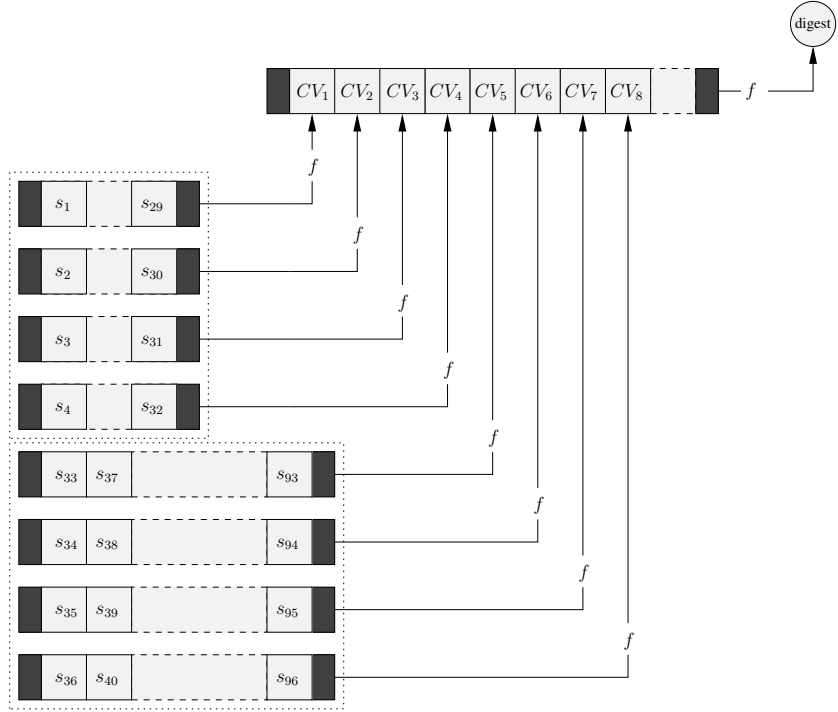}
}
\caption{Variant of Mode~4L (with a tree height~$2$) for the support of slice interleaving with a Multi-Core SIMD architecture. 
At the base level, only the first two groups are presented.
In a multithreaded implementation with a fixed thread pool, the group computations are 
distributed among the threads, but each thread is responsible for processing the nodes of a same group in their entirety, by means of a SIMD implementation.}
\label{Mode_4L_interleaving_example}
\end{figure}

\section{Concluding remarks}\label{concl}

In this paper, we introduced different ways of constructing a SHAKE function based on an inner function RawSHAKE and a parameterizable tree of nodes.
Resulting modes are \textsc{Sakura}-compatible and lead to interesting asymptotic bounds with regard to the memory usage in a sequential implementation, the 
ideal parallel time and the number of involed processors.
We showed new directions for tree hash modes with a focus on space-time trade-offs, whether these modes are intended for stored content or 
for streaming live content.
We think that it would be of interest to propose variants of our proposed modes using various optimizations
suitable to software implementations. 


\paragraph*{On the accuracy of complexity measures.} 
This paper has focused on the asymptotic analysis (using big-O) of tree hash modes but there is also an interest for their optimizations at finite distance 
(\emph{i.e.} by giving exact complexity results) in order to produce
efficient tree hash circuits. For instance, in the case of Mode 3,
an interesting work is to reduce at best the number of rounds\footnote{The running time of a round is one unit of time, that is, one call to the permutation.} 
of parallel calls to the permutation (\emph{i.e.} the depth of the circuit, which is in $O(\log n)$)
by taking into account the frame bits overheads of \textsc{Sakura}. A much more interesting challenge is to optimize the mode in depth and width: the goal is to
minimize the maximum number of parallel calls to the permutation that can occur in one round, subject to the constraint of a minimal number of rounds.
This kind of work was initiated in~\cite{AtRo17} with the use of a minimal coding for the nodes.

\paragraph*{Exploiting parallelism for the output extension.}
Besides, we supposed the result of the SHAKE function to be of fixed size  (512 bits), while its main functionality is to generate a hash of arbitrary length. 
If we nevertheless decide to allow this functionality by using the ``squeezing'' phase of the last evaluation of RawSHAKE (on the final node), 
our asymptotic results lose their relevance.
Indeed, the squeezing phase being purely sequential, if it is decided to generate a hash of, say, $n$ bits from a message of length $n$, 
the total number of evaluations of the permutation is then in $O(n)$, no matter the type of tree used to compress the message.
To remedy this problem, that one ``squeezing'' phase can be demultiplied using an inverted tree of type GGM (Goldreich-Goldwasser-Micali \cite{GGM84}) 
having for root the final node. 
In this configuration, there are two trees: one to compress the message and the other one to expand the hash to the desired length.
Nodes of the second tree are defined such that each node is a part of an $f$-image except for the root node which is an $f$-input.
Originally, GGM construction allows a long pseudo-random bitstring to be generated from a short (pseudo-)random seed, by means of a binary tree.
In our case, the used tree is not necessarily binary. Instead, its topology coherently depends on the parameters used in the first tree
with, nevertheless, a slight difference if the node arities are not fixed per level: For each node, the child node arities are in reverse order as compared 
to the first tree. Let us take an example where the digest of a 9-block message is computed thanks to Mode~4L with a tree height $h=2$ and we require this digest to be
of same length as the message. In the first tree, the base level node arities will be 2, 3 and 4 while they will be 4, 3 and 2 in the second tree.
More precisely, the arity of a node determines the length of the corresponding hash generated by RawSHAKE.
Indeed, it seems that these parameters can be reused in this second inverted\footnote{It is inverted in the sense that the flows of computation are from root to leaves.} 
tree so as to maintain the asymptotic parallel time. 
A GGM tree normally has all its leaves at the same depth. Its height is then chosen such that the total number of bits
of the leaves is greater than or equal to the desired hash length. In order to generate a hash of a given size, some of its branches should be pruned and 
one leaf should probably be truncated. 
We finally note that interleaving techniques could be used in this second tree to take better advantage of
SIMD instructions. We think that these last aspects deserve particular attention. 
If we do not want to use a GGM tree construction to expand the hash, an alternative is a Counter based PRG, \emph{e.g.} CTR-DRBG (Counter mode Deterministic 
Random Byte Generator \cite{BK12}) whose ideal parallel running time is constant.

\paragraph*{Benefits of using additional operators from the parallelism standpoint.} As a last remark, there were no modes referenced in our tables (Section~\ref{survey} 
and \ref{parameterizable_mode}) 
having an asymptotically sublogarithmic (or even constant) parallel running time. 
The reason is that using only $f$ as operator, such complexities are not possible. In order to make a parallel running time 
in $O(1)$ possible, the hashing mode should require a constant number of rounds of parallel calls to the underlying permutation, and the intermediate results of the last round
should be combined with a second operator of negligible cost (\emph{e.g.} the XOR operator). The \emph{Randomize-then-Combine} paradigm has been proposed by Bellare 
\emph{et al.} \cite{BM97}, with several examples of combining operators. The XOR operator must be avoided for arbitrary-length messages, 
due to an attack based on Gauss elimination. More generally, some of the proposed operators in finite groups are not resistant to the subexponential 
\emph{Generalized Birthday Attack} \cite{Wag02}. The possibility of such a scheme nevertheless looks interesting. 
As we have seen previously, the use of a traditional tree hash mode (using only $f$) by a set of computers having
different computational resources leads to a race to the bottom with respect to the potential parallelism.
Indeed, when transmitting a signed message, the parameters of the tree topology chosen by the transmitting computer are not necessarily optimal for the receiving one.
A hashing scheme with a $O(1)$ ideal parallel running time resolves this problem. Besides, if the combining operator is both associative and commutative, 
the distribution of work among the computing units can be done in any manner, without any tree topology constraint.

\bibliographystyle{plain}
\bibliography{trees} 

\section*{Appendix}

\appendix

\section{Additional parameters for live-streamed content}\label{appendix_param}

\paragraph*{Mode WC (Weak Compromise).} We rework the tree structure of Mode~3
by fixing the arity of the nodes to a constant $k$ which is a power of $2$. For instance, if $k=2^4$, the number of hash states 
of a sequential implementation is approximately divided by $4$, but in return, the parallel running time is approximately multiplied by $4$. 
Generally, if $k=2^i$ where $i$ is a positive integer, then the memory consumption of a sequential implementation is approximately divided by $i$ while the parallel
running time is multiplied by $k/\log_2 k$. There is no 
change asymptotically (using Big-O).~\\

One way to reduce the memory consumption is to use an increasing sequence $(u_i)_{i \geq 1}$ of arities, where $u_i$ is the arity of level $i$. We thus seek
a sequence $(u_i)_{i \geq 1}$ meeting the constraint $\prod_{i=1}^h u_i \geq n$ where the integer $h$ and the sum $\sum_{i=1}^h u_i$ reflect
the desired complexities, in terms of memory usage and parallel time.
First, we note that for this kind of tree it is impossible to adapt Mode 4S to support live-streamed content while keeping the same complexities.
Suppose that the tree height $h$ is constant. It is indeed impossible to have a finite sequence $(u_i)_{i=1 \ldots h}$ meeting
this constraint for all possible message size $n$, and where the terms $u_i$ do not depend on $n$.
However, in what follows, we show that certain sequences of arities give interesting asymptotic efficiencies.

\paragraph{Mode B1.} The arities $u_i$ of the levels $i \geq 1$ grow as follows:
$$u_i = 2^i\ \forall i \geq 1.$$

\begin{theorem}
Suppose that we have as many processors as we want. Mode~B1 has $O\left(n^{\frac{1}{\sqrt{2 log n}}}\right)$ parallel running time and requires
$O\left(\sqrt{\log n}\right)$ hash states in memory to process the message with a single processor.
\end{theorem}

\begin{proof}
Asymptotically, the amount of memory consumed in a sequential implementation
corresponds, in the worst case, to the height $k$ of the tree.
We seek the lowest $k$ such that $\prod_{i=1}^k u_i\geq n$.
 We have $\prod_{i=1}^{k} 2^i \geq n$ which is equivalent to $\sum_{i=1}^k i \geq \frac{\log n}{\log 2}$. Since this sum is positive and polynomial in $k$, 
 we have $\frac{k^2}{2} \sim \frac{\log n}{\log 2}$, \emph{i.e.} $k=O(\sqrt{2\log})$. The parallel time is an affine function of 
 $\sum_{i=1}^k 2^i \sim 2^{k} = O\left(n^{\frac{1}{\sqrt{2 \log_2 n}}}\right)$.
\end{proof}

\paragraph{Mode B2.} The arities $u_i$ of the levels $i \geq 1$ grow as follows:
$$u_i = i+1\ \forall i \geq 1.$$

\begin{theorem}
Suppose that we have as many processors as we want. Mode~B2 has $O\left(\frac{\log^2 n}{\log^2 \log n}\right)$ parallel running time and requires
$O\left(\frac{\log n}{\log \log n}\right)$ hash states in memory to process the message with a single processor.
\end{theorem}

\begin{proof}
Asymptotically, the amount of memory consumed in a sequential implementation
corresponds, in the worst case, to the height $k$ of the tree.
We seek the lowest $k$ such that $\prod_{i=1}^k u_i\geq n$.
It has been shown before that $\log( \prod_{i=1}^k u_i) \sim \log n$.
According to the Stirling Formula, $k! \sim k^ke^{-k}\sqrt{2\pi k}$. Thus, we have
$\log n\sim \log(k^ke^{-k}\sqrt{2\pi k})=k\log k -k+\frac{1}{2}\log(2\pi k)\sim k\log k$. Hence,
$\log\log n \sim \log (k\log k)=\log k + \log\log k\sim \log k$. We finally obtain
$k\sim\dfrac{\log n}{\log k}\sim \dfrac{\log n}{\log\log n}$. If we have as many processors as we want, then
the amount of blocks processed sequentially is asymptotically the sum of the $k$ level arities. Considering that $\sum_{i=2}^k i \sim \frac{k^2}{2}$,
we deduce the expected result.

\end{proof}

\paragraph{Mode B3.} The arities $u_i$ of the levels $i \geq 1$ grow as follows:
$$u_i =  \left\lfloor \log (i+3) \right\rfloor\ \forall i \geq 1.$$

\begin{theorem}
Suppose that we have as many processors as we want. Mode~B3 has $O\left(\frac{\log n \log \log n}{\log \log \log n}\right)$ parallel running time and requires
$O\left(\frac{\log n}{\log \log \log n}\right)$ hash states in memory to process the message with a single processor.
\end{theorem}

\begin{proof}
For the sake of simplification, we can seek the lowest $k$ such that \break $\prod_{i=2}^k \log i \geq n$, which is
an equivalent problem. Applying the $\log$ function to both sides, we can show that $\sum_{i=2}^k \log \log i \sim \log n$. 
Asymptotically, we can then seek $k$ such that $k \log \log k \sim \sum_{i=2}^k \log \log i \sim \log n$.
The solution $k$ is such that $k \sim \frac{\log n}{\log \log \log n}$.
Regarding the parallel running time, it follows that
$\sum_{i=2}^k \log i \sim k \log k \sim \frac{\log n}{\log \log\log n} \log \left(\frac{\log n}{\log \log\log n}\right)$, 
giving the expected result.
\end{proof}

Note that this paper does not investigate sequences behaving like superfactorial or hyperfactorial, which probably lead to other 
interesting trade-offs.

As for the modes dedicated to stored content, we can reduce the number of processors allowing the ideal parallel running time to be reached 
by means of a rescheduling technique. Indeed,
since the processors have more computations to do at higher levels of the tree (and the lower level nodes have smaller arity), we could assign to each 
processor a subtree that represents the amount of computations of a higher level node. In doing so, the parallel running time is only increased by a constant factor.
This aspect will be developed in future works. 

\end{document}